\documentclass[a4paper,UKenglish,cleveref, autoref, thm-restate,numberwithinsect]{lipics-v2021}
\usepackage{nicefrac}
\usepackage{xspace}
\usepackage[dvipsnames,table]{xcolor}
\usepackage{cleveref}
\usepackage{graphicx}
\usepackage{multirow}
\usepackage{algorithm,algpseudocode,amsmath}
\makeatletter
\newcounter{phase}[algorithm]
\newlength{\phaserulewidth}
\newcommand{\setphaserulewidth}{\setlength{\phaserulewidth}}
\newcommand{\phase}[1]{%
  \vspace{-1.25ex}
  \Statex\leavevmode\llap{\rule{\dimexpr\labelwidth+\labelsep}{\phaserulewidth}}\rule{\linewidth}{\phaserulewidth}
  \Statex\strut\refstepcounter{phase}\textit{Phase~\thephase~--~#1}% Phase text
  \vspace{-1.25ex}\Statex\leavevmode\llap{\rule{\dimexpr\labelwidth+\labelsep}{\phaserulewidth}}\rule{\linewidth}{\phaserulewidth}}
\makeatother

\setphaserulewidth{.7pt}

\graphicspath{{./figures/}} %helpful if your graphic files are in another directory

\title{Scheduling with Machine Conflicts} %TODO Please add

\titlerunning{Scheduling with Machine Conflicts} %TODO optional, please use if title is longer than one line

\author{Moritz Buchem}{School of Business and Economics, Maastricht University, The Netherlands }{m.buchem@maastrichtuniversity.nl}{https://orcid.org/0000-0002-1590-346X}{}

\author{Linda Kleist}{Department of Computer Science, TU Braunschweig, Germany}{l.kleist@tu-bs.de}{https://orcid.org/0000-0002-3786-916X}{}

\author{Daniel Schmidt genannt Waldschmidt}{Institut f\"ur Mathematik, TU Berlin, Germany}{dschmidt@math.tu-berlin.de}{https://orcid.org/0000-0002-9331-445X}{}

\authorrunning{M.\,Buchem, L.\,Kleist and D.\,Schmidt genannt Waldschmidt} %TODO mandatory. First: Use abbreviated first/middle names. Second (only in severe cases): Use first author plus 'et al.'

\Copyright{Moritz Buchem, Linda Kleist and Daniel Schmidt genannt Waldschmidt} %TODO mandatory, please use full first names. LIPIcs license is "CC-BY";  http://creativecommons.org/licenses/by/3.0/

\funding{Daniel Schmidt genannt Waldschmidt was funded by the DFG under Germany's Excellence Strategy – The Berlin Mathematics Research Center MATH+ (EXC-2046/1, project ID: 390685689).}
\acknowledgements{
	We  thank Nicole Megow for an inspiring discussion concerning further applications and 
 Tjark Vredeveld for constructive remarks concerning the presentation of this manuscript. 
}%optional

\nolinenumbers %uncomment to disable line numbering

\hideLIPIcs  %uncomment to remove references to LIPIcs series (logo, DOI, ...), e.g. when preparing a pre-final version to be uploaded to arXiv or another public repository

\EventEditors{John Q. Open and Joan R. Access}
\EventNoEds{2}
\EventLongTitle{42nd Conference on Very Important Topics (CVIT 2016)}
\EventShortTitle{CVIT 2016}
\EventAcronym{CVIT}
\EventYear{2016}
\EventDate{December 24--27, 2016}
\EventLocation{Little Whinging, United Kingdom}
\EventLogo{}
\SeriesVolume{42}
\ArticleNo{196}
\newcommand{\nonNested}{basic\xspace}

\newcommand{\I}{\ensuremath{\mathcal{I}}\xspace}

\newcommand{\qtime}{system time\xspace}

\newcommand{\ScM}{\textsc{smc}$(G,\J)$\xspace}
\newcommand{\ScMP}{\textsc{smc}\xspace}

\newcommand{\ScMPropiP}{\textsc{smc-prop-i}\xspace}

\newcommand{\ScMPropiiP}{\textsc{smc-prop-ii}\xspace}

\newcommand{\ScMPropiiiP}{\textsc{smc-prop-iii}\xspace}

\newcommand{\ScMIDSh}{\textsc{smc-id-short}$(G,n,(\lb{},p,\rb{}))$\xspace}

\newcommand{\IdenticalP}{\textsc{smc-Id}\xspace}

\newcommand{\Unit}{\textsc{smc-Unit}$(G,n)$\xspace}
\newcommand{\UnitP}{\textsc{smc-Unit}\xspace}

\newcommand{\OPT}{\textsc{opt}\xspace}

\newcommand{\lb}[1]{\ensuremath{\harpR{b_{#1}}}\xspace}
\newcommand{\rb}[1]{\ensuremath{\harp{b_{#1}}}\xspace} 

\newcommand{\J}{\ensuremath{\mathcal{J}}\xspace}

\newcommand{\jcm}{\textsc{jc-intra}\xspace}
\newcommand{\jct}{\textsc{jc-inter}\xspace}

\def\NP{\ensuremath{\text{NP}}\xspace}
\def\P{\ensuremath{\text{P}}\xspace}

\newcommand{\harp}[1]{\mathpalette\harpoonvec{#1}}

\newcommand{\harpR}[1]{\mathpalette\harpoonvecREV{#1}}

\newcommand{\harpvecsignREV}{\scriptscriptstyle{\leftharpoonup}}

\newcommand{\harpoonvecREV}[2]{%
	\ifx\displaystyle#1\doalign{$\harpvecsignREV$}{#1#2}\fi
	\ifx\textstyle#1\doalign{$\harpvecsignREV$}{#1#2}\fi
	\ifx\scriptstyle#1\doalign{\scalebox{.6}[.9]{$\harpvecsignREV$}}{#1#2}\fi
	\ifx\scriptscriptstyle#1\doalign{\scalebox{.5}[.8]{$\harpvecsignREV$}}{#1#2}\fi
}
\newcommand{\harpvecsign}{\scriptscriptstyle\rightharpoonup}
\newcommand{\harpoonvec}[2]{%
	\ifx\displaystyle#1\doalign{$\harpvecsign$}{#1#2}\fi
	\ifx\textstyle#1\doalign{$\harpvecsign$}{#1#2}\fi
	\ifx\scriptstyle#1\doalign{\scalebox{.6}[.9]{$\harpvecsign$}}{#1#2}\fi
	\ifx\scriptscriptstyle#1\doalign{\scalebox{.5}[.8]{$\harpvecsign$}}{#1#2}\fi
}
\newcommand{\doalign}[2]{%
	{\vbox{\offinterlineskip\ialign{\hfil##\hfil\cr#1\cr$#2$\cr}}}%
}

\Crefname{observation}{Obs}{Obs}
\Crefname{theorem}{Thm}{Thm}
\Crefname{corollary}{Cor}{Cor}
\crefname{observation}{Observation}{Observations}
\crefname{theorem}{Theorem}{Theorems}
\crefname{corollary}{Corollary}{Corollaries}

\begin{document}

\maketitle

\begin{abstract}
We study the scheduling problem of makespan minimization while taking machine conflicts into account. Machine conflicts arise in various settings, e.g., shared resources for pre- and post-processing of tasks or spatial restrictions.

In this context, each job has a blocking time before and after its processing time, i.e., three parameters. We seek for conflict-free schedules in which the blocking times of no two jobs intersect on conflicting machines. Given a set of jobs, a set of  machines, and a graph representing  machine conflicts, the problem \textsc{SchedulingWithMachineConflicts}~(\ScMP), asks for a conflict-free schedule of minimum makespan.

We show that, unless $\P=\NP$, \ScMP on $m$ machines does not allow for a $\mathcal{O}(m^{1-\varepsilon})$-approximation algorithm  for any $\varepsilon>0$, even in the case of identical jobs and every choice of fixed positive parameters, including the unit case. 
Complementary, we provide approximation algorithms when a suitable collection of independent sets is given. Finally, we present polynomial time algorithms to solve the problem for the case of unit jobs on special graph classes. Most prominently, we  solve it for bipartite graphs by using structural insights for  conflict graphs of star forests.

\end{abstract}

\section{Introduction}
%%%%%%%%%%%%%%%%%%%%%%%%  MOTIVATION AND APPLICATION  %%%%%%%%%%%%%%%%%%%%%
Distributing tasks smartly is a challenge we face in numerous settings, ranging from every day life to optimization of industrial processes. Often these assignments must satisfy additional requirements. Such requirements may result from using shared resources and prohibit simultaneous processing of certain tasks. Such conflicts may arise in diverse settings. Surprisingly, and to the best of our knowledge, scheduling with non-trivial machine conflicts has not been investigated previously.

In this work, we generalize the well-studied scheduling problem of makespan minimization on identical parallel machines by taking machine conflicts into account. We are particularly interested in the situation when external pre- and post-processing of jobs is necessary before and after the job is internally processed by a machine. Such pre- and post-processing may be due to shared resources or spatial constraints. 

An example of the first setting can be found in computing problems, in which processors may share different databases or external processors that must be accessed before and after executing tasks on the processor. These external resources can only be accessed by one processor at a time and different processors may share different external resources. 

Another example can be found in logistic and production processes, e.g., when several machines are served by a common robot arm, the transit times should clearly not overlap. An up-to-date example of spatial conflicts occurs in pandemics when schedulers are faced with potentially infectious jobs which should keep sufficient distance to each other, e.g., in testing or vaccination centers. Similarly, spatial  conflicts play a crucial role when jobs may have private information or data that should not be shared; e.g.,  the interrogation of suspects in multiple rooms.

%%%%%%%%%%%%%%%%%%%%%%%%  PROBLEM FOMULATION  %%%%%%%%%%%%%%%%%%%%%%%%%%%%%

\subparagraph{The problem.}
 \textsc{SchedulingWithMachineConflicts} (\ScMP) is a scheduling problem in which jobs on conflicting machines are processed such that certain blocking intervals of their processing time do not overlap. The objective is to minimize the makespan.

An instance of $\ScMP$ is defined by a set of jobs $\mathcal{J}$ and a conflict graph $G=(V,E)$ on a set of machines $V$ where  two machines~$i$ and $i'$ are \emph{in conflict} if and only if $\{i,i'\} \in E$. 
In contrast to classical scheduling problems, each job $j$ has three parameters~$(\lb{j},p_j,\rb{j})$, where $\lb{j}$ and $\rb{j}$ denote the first and second \emph{blocking time} of $j$, respectively, and $p_j$ denotes its \emph{processing time}. Together they constitute the \emph{\qtime} $q_j=\lb{j}+p_j+\rb{j}$; note that the order $\lb{j},p_j,\rb{j}$ must be maintained.
We seek schedules in which the blocking times of no two jobs on conflicting machines intersect.
Formally, a \emph{(conflict-free) schedule} $\Pi$ is an assignment of jobs to machines and starting times such that
\begin{itemize}
	\item  for each point in time, every machine executes at most one job,
	\item for every edge $\{i,i'\} \in E$ and two jobs $j,j'\in \J$ assigned to
	machines $i$ and $i'$, respectively,  the intervals of the  blocking times of $j$ and $j'$ do not overlap interiorly in time.
\end{itemize}
In particular, every job $j$ has a starting time $S_j^\Pi$, and a completion time  $C_j^\Pi=S_j^\Pi+q_j$. We say a job is \emph{running} at every point in time in the open interval $(S_j^\Pi,C_j^\Pi)$.
 We define the \emph{makespan} of $\Pi$ as $\|\Pi\|:=\max_{j\in\J}C_j^\Pi$ and seek for a schedule with minimum makespan, i.e., the objective is
$\min_\Pi \|\Pi\|$.

Throughout this paper, we use $n:=|\mathcal{J}|$ and $m:=|V|$ to refer to the number of jobs and machines, respectively. 
We highlight that, by definition, all schedules are non-pre\-emptive.

\subsection{Our Contribution and Organization}
We formalize the concept of machine conflicts by introducing a conflict graph. The problem \ScMP generalizes the classical problem $P||C_{\max}$ of makespan minimization on parallel identical machines and models various situations in which machines may not process certain parts of jobs in parallel.

In \cref{sec:Non-nested}, we consider instances with  \emph{long blocking times}.  In this case, \ScMP reduces to identifying a maximum independent set of the conflict graph in a first step and then minimizing the makespan on the identified machines in a second step. This connection to the maximum independent set problem implies that there exists no approximation algorithm for \ScMP on general graphs
  with a performance guarantee in $O(m^{1-\varepsilon})$ for any $\varepsilon>0$ (\cref{theorem:hardness_identical_long_blocking_time}). 
If a maximum or approximate maximum independent set for a graph $G$ is given, approximation algorithms of $P||C_{\max}$ can be exploited to find exact or approximate solutions for \ScMP on $G$ and jobs with long blocking times (\cref{cor:nonNested_polytime,theorem:ApprIndepSetLong}).

In \cref{sec:IdenticalNested}, we strengthen these insights by linking \ScMP with identical jobs (\IdenticalP) for all fixed parameters and \emph{short blocking times} to the problem of finding an appropriate collection of disjoint independent sets. Similar to the case of long blocking times, we provide inapproximability results for \IdenticalP with short blocking times (\cref{theorem:hardness_identical_short_blocking_time}) and approximation algorithms when such a collection of independent sets is at hand (\cref{thm:approximation_identical_short_max,cor:approximation_identical_short_approxmax}).

 In \Cref{sec:unit}, we consider \ScMP with unit jobs (\UnitP), i.e., $\lb{}=p=\rb{}=1$. Motivated by  the inapproximability result for \UnitP on general graphs by  \cref{theorem:hardness_identical_short_blocking_time}, we focus on special graph classes.
For complete graphs, we show that the problem can be reduced to using two machines and hence can be solved efficiently (\cref{lem:UNITcomplete}).
Most interestingly, we present a polynomial time algorithm to compute optimal schedules on bipartite graphs. 
Bipartite graphs are of special interest, because for every point in time, the set of active machines induces a bipartite graph. Hence, our insights can be understood as   \emph{local optimality criteria} of schedules for all graphs.
Our results are based on structural insights for stars and identifying a spanning star forest with special properties for each bipartite graph. 
All algorithms have a running time polynomial in the size of $G$ and $\log(n)$.

%%%%%%%%%%%%%%%%%%%%%%%%  RELATED WORK  %%%%%%%%%%%%%%%%%%%%%%%%%%%%%

\subsection{Related Work}
Our problem \ScMP generalizes the classical scheduling problem of \emph{makespan minimization on parallel identical machines}, also denoted by $P||C_{\max}$.
There are two possible options to view this generalization because $P||C_{\max}$ is equivalent to $\ScMP$ 
\begin{itemize}
	\item if the blocking times of all jobs vanish, i.e.,  $\lb{j}=\rb{j}=0$ for all $j\in\J$, or
	\item  if the edge set of the conflict graph is the empty set.
\end{itemize}
For a constant number of machines, $P||C_{\max}$ is weakly NP-hard, while it is strongly NP-hard when $m$ is part of the input~\cite{GJ79}. Graham~\cite{G66,G69} introduced list scheduling algorithms to obtain the first constant approximation algorithms for this problem. Improved approximation guarantees have been achieved by a fully polynomial time approximation scheme (FPTAS) when $m$ is constant~\cite{Sah76} and a polynomial time approximation scheme (PTAS) when $m$ is part of the input~\cite{HS87}. In subsequent work, the latter has been improved to efficient polynomial time approximation schemes (EPTAS), we refer to~\cite{AAWY98,CJZ13,Hoch97,Jan10,JKV20}. 
\medskip

\emph{Scheduling with conflict graphs for jobs} has been investigated when conflicts between jobs are specified. To this end,  one distinguishes two types of restrictions: 
\begin{enumerate}
	\item conflicting jobs cannot be scheduled on the same machine (\jcm), or 
	\item conflicting jobs cannot be processed concurrently on different machines (\jct).
\end{enumerate}
For \jcm, Bodlaender and Jansen~\cite{BJ93} 
show 
NP-hardness even for unit time jobs on  bipartite graphs and co-graphs.
Bodlaender et al.~\cite{BJW94} present approximation algorithms for  bipartite graphs and graphs with bounded treewidth. 
Das and Wiese~\cite{DW17} provide a PTAS for identical machines where the conflict graph is a collection of cliques.

\jct has been studied under various names: mutual exclusion scheduling problem (MES),  scheduling with conflicts and scheduling with agreements where an agreement graph is the complement of a conflict graph.
Baker and Coffman~\cite{BC96} showed that MES is NP-hard for general graphs and a fixed number of $m\geq 3$ machines. Furthermore, Even et al.~\cite{EHKR09} showed that it is even APX-hard. MES is closely related to finding a proper vertex coloring of $G$ with the minimum number of colors such that each color appears at most $m$ times. The computational complexity of MES has been investigated for various graph classes, see~\cite{BC96, BF05,BJ95, dW97,G09, H93}. 
For bipartite graphs, MES remains NP-hard~\cite{BJ95}, while it becomes polynomial time solvable for a fixed number of machines~\cite{BC96,BJ95, H93}. Even et al.~\cite{EHKR09} consider \jct for non unit jobs on two machines with few job types and show that it can be solved (approximately) when $p_{j} \in \{1,2,3\}$. When $p_{j} \in \{1,2,3,4\}$, \jct is APX-hard even for bipartite conflict graphs~\cite{EHKR09}. For complements of bipartite graphs, Bendraouche and Boudhar~\cite{BB12} showed NP-hardness for a restricted number of job types. Mohabeddine and Boudhar~\cite{MB19} showed that \jct is NP-hard on complements of trees, while it is solvable in polynomial time for complements of caterpillars or cycles.
 \medskip
 
 \emph{Scheduling with pre- and post-processing} has been considered in different models in the literature. Two of these models are the master-and-slave problem introduced by Kern and Nawijn~\cite{KN91} and termed by Sahni~\cite{Sah96} and the problem of scheduling jobs with segmented self-suspension with a single suspension component introduced by Rajkumar et al.~\cite{RSL88}. In these problems pre- and post-processing operations must be executed on (master) machines in the correct order per job with a necessary amount of time in between these operations spent either on a slave machine or in self-suspension. Chen et al.~\cite{CHHMB19} discuss the relation between these two problems and present old and new approximation results for different special cases.

\section{Long Blocking Times and the Maximum Independent Set Problem}\label{sec:Non-nested}
In this section, we identify several special cases of \ScMP which reduce to the problem of (classical) makespan minimization on a maximum independent set of identical machines. The common denominator of these cases is the property that each job has a "long blocking time". These long blocking times lead to so-called \emph{\nonNested} schedules in which jobs are on conflicting machines do not run in parallel. To this end, we introduce some notation in a general form in view of the next sections.

\subparagraph{(Maximum) induced $c$-colorable subgraph.}
For a graph $G=(V,E)$ and $c\in\mathbb{N}_{\geq 1}$,
a \emph{(maximum) induced $c$-colorable subgraph}, or short \emph{(maximum) $c$-IS}, of $G$ is a set of $c$ disjoint independent sets $\I_1,\ldots,\I_c \subseteq V$ (whose union has maximum cardinality). We denote the cardinality of a maximum $c$-IS of $G$ by $\alpha_c(G)$. Note that we also write $\alpha_c$ if $G$ is clear from the context.
Clearly, a $1$-IS is an independent set. It is well known that, unless $P=NP$, there exists no $\mathcal{O}(m^{\epsilon-1})$-approximation for any $\epsilon > 0$ for finding a maximum $1$-IS on a graph with $m$ vertices~\cite{Haa99, Z06}. Lund and Yannakakis~\cite{LY93} show that the same inapproximability result can be extended to finding maximum $c$-IS{s} for any $c\in\mathbb{N}_{\geq 1}$. Furthermore, we define a class of schedules that we use later in this section.

\subparagraph{Basic schedules.}
A schedule $\Pi$ is \emph{\nonNested}, if 
for every edge $ii'\in E$ and
for every pair of jobs $ j$ and $j'$ assigned to $i$ and $i'$, respectively, 
their system times are non-overlapping, i.e., 
 $S_j^{\Pi}\leq S_{j'}^{\Pi}$  implies
$ C_j^{\Pi} \leq S_{j'}^{\Pi}$.

\smallskip
We identify three types of instances, in which all feasible schedules are \nonNested.
\begin{restatable}{lemma}{nonNestedLemma}\label{lemma:nonNested}
	For an instance of $\ScM$, every schedule $\Pi$ is \nonNested if any of the three following properties are fulfilled:
	\begin{romanenumerate}
		\item all jobs $j$ are identical $(\lb{j}=\lb{},p_j=p,\rb{j}=\rb{})$ and $(\max\{\lb{},\rb{}\}>p)$, \hfill(\textsc{prop-i})
		\item all jobs $j$ have equal \qtime $(q_j=q)$ and $\lb{j}>p_j$ and $\rb{j}>p_j$, or  \hfill(\textsc{prop-ii})
		\item all jobs $j,j'$ fulfill $\lb{j}>p_{j'}$ and $\rb{j}>p_{j'}$. \hfill(\textsc{prop-iii})
	\end{romanenumerate}
\end{restatable}

\begin{proof}
	
Suppose for the sake of a contradiction that there exists two jobs $j$ and $j'$ assigned to machines $i$ and $i'$, respectively, with $ii'\in E$ such that  $j'$ starts while $j$ is processed.
Because their blocking times cannot overlap, we also know that $j'$ starts after the first blocking time of $j$
and the first blocking time of $j'$ ends before the second blocking time of $j$ starts, see \cref{fig:nonNested}. 	Consequently, $\lb{j'}\leq p_j$.
Moreover, $p_{j'}\geq \rb{j}$ if $j'$ ends after $j$; in particular, this holds for equal \qtime{s} as in (i) and (ii).
	
	%%%%%%%%%%%%%%%%%
	\begin{figure}[hb]
		\centering
		\includegraphics[page=4]{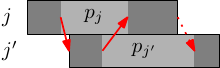}
		\caption{Illustration of the proof of \cref{lemma:nonNested}.}
		\label{fig:nonNested}
	\end{figure}
	%%%%%%%%%%%%%%

	\begin{romanenumerate}
		\item Because all jobs have the same parameters, we obtain the contradiction $\lb{}=\lb{j'}\leq p_j=p$ and $p=p_{j'}\geq \rb{j}=\rb{}$.
		\item Using the assumptions $\lb{j'}>p_{j'}$ and $\rb{j}>p_j$,
		we obtain  $\lb{j'}>p_{j'}\geq \rb{j}>p_j\geq \lb{j'}$, a contradiction.
	\item The assumption $\lb{j'}>p_j$ yields an immediate contradiction to $\lb{j'}\leq p_j$.\qedhere 
	\end{romanenumerate}
\end{proof}

We denote \ScMP restricted to instances fulfilling one of the above three properties by \ScMPropiP, \ScMPropiiP and \ScMPropiiiP, respectively. As it turns out, even in the case of identical jobs, these problems are hard to solve since the problem of computing a maximum $1$-IS reduces to them.

\begin{restatable}{theorem}{HardnessIdenticalLongBlockingTime}\label{theorem:hardness_identical_long_blocking_time}
	Unless $\P=\NP$,  \ScMPropiP, \ScMPropiiP and \ScMPropiiiP do not admit a $\mathcal{O}(m^{1-\varepsilon})$-approximation for any $\varepsilon > 0$, even in the case of identical jobs and when $\lb{},p,\rb{}$ are fixed.
\end{restatable}

\begin{proof}
	By containment of \ScMPropiiP in \ScMPropiiiP, it suffices to show the inapproximability for \ScMPropiP and \ScMPropiiP. We restrict our attention to instances consisting of $n$ jobs with $q=\lb{} + p +\rb{} = 1$ and $n=\alpha_1(G)$. By~\cref{lemma:nonNested}, schedules for \ScMPropiP or \ScMPropiiP are \nonNested. Clearly, the optimum makespan \OPT is $1$ because $q = 1$.
	Suppose for some constant $\kappa > 0$ there exists a $(\kappa m^{1-\varepsilon})$-approximation algorithm for $\varepsilon > 0$ and let $\Pi$ denote its schedule.
	Moreover, let $\beta$ denote the maximum number of machines processing jobs in parallel in $\Pi$ at any point in time. As we have at most $n$ distinct starting times in $\Pi$ we can compute $\beta$ in polynomial time.
	Because $\Pi$ is \nonNested, we can transform $\Pi$ into a new schedule $\Pi'$ that processes all jobs non-idling on $\beta$ machines without increasing the makespan. Hence, we obtain
	$
	\|\Pi\|\geq \left\lceil\nicefrac{n}{\beta}\right\rceil 
	\geq \nicefrac{n}{\beta}
	$.
	This fact together with the assumption $\|\Pi\|\leq \kappa m^{1-\varepsilon} \cdot\OPT=\kappa m^{1-\varepsilon}$ yields
	$
	\beta\geq\nicefrac{n}{\|\Pi\|}
	\geq \nicefrac{n}{\kappa m^{1-\varepsilon}}=\nicefrac{1}{\kappa m^{1-\varepsilon}}\cdot\alpha(G)$.
	In other words, the $(\kappa m^{1-\varepsilon})$-approximation algorithm implies an $(\nicefrac{1}{\kappa}\cdot m^{\varepsilon-1})$-approximation algorithm for computing a maximum $1$-IS for every graph $G$; a contradiction~\cite{Haa99, Z06}.
\end{proof}

On the contrary, when a maximum $1$-IS of the conflict graph is at hand, we can compute optimal and near-optimal schedules for the above problems in polynomial time as it reduces to the classical problem $P||C_{\max}$.

\begin{restatable}{theorem}{nonNestedPolytime}\label{cor:nonNested_polytime}
	If a maximum $1$-IS of the conflict graph is given, then 
	\begin{romanenumerate}
		\item 
		\ScMPropiP \label{cor:nonNested_polytimei}
		can be solved in polynomial time, 
		\item
		\ScMPropiiP can be solved in polynomial time, and \label{cor:nonNested_polytimeii}
		\item 
		\ScMPropiiiP has a PTAS. \label{cor:nonNested_polytimeiii}
	\end{romanenumerate}
\end{restatable}

\begin{proof}
	Let $\Pi$ be an optimal schedule for one of the three problems. By~\cref{lemma:nonNested} $\Pi$ is \nonNested, hence, at any point in time at most $\alpha_1(G)$ jobs are running. Therefore, we can modify the job-to-machine assignment of $\Pi$ such that all jobs are processed on a maximum $1$-IS while maintaining the starting times of the jobs. As we did not increase the makespan, we have an optimal schedule where all jobs are assigned to a maximum $1$-IS and hence, the problem reduces to $P||C_{\max}$ on $\alpha_1(G)$ machines. Thus, evenly distributing all jobs with equal system time over the machines (for \ScMPropiP and \ScMPropiiP) and an already existing PTAS for $P||C_{\max}$ (for \ScMPropiiiP) yields the desired result.
\end{proof}

Note that if the number of machines is constant, i.e. $m = O(1)$, the results above hold via complete enumeration for finding a maximum $1$-IS. For the third case, this improves to an FPTAS by~\cite{Sah76}.
Similarly, when a $\nicefrac{1}{\gamma}$-approximate $1$-IS is at hand, we obtain approximation results using a list scheduling approach on the independent set.

\begin{restatable}{theorem}{ApprIndepSetLong}\label{theorem:ApprIndepSetLong}
	Given a $\nicefrac{1}{\gamma}$-approximate $1$-IS of the conflict graph, then there exists
	\begin{romanenumerate}
		\item a $\lceil\gamma\rceil$-approximation for \ScMPropiP, \label{item:approxi}
		\item a $\lceil\gamma\rceil$-approximation for \ScMPropiiP, and
		\item a $\left(\gamma + 1 - \nicefrac{1}{m}\right)$-approximation for \ScMPropiiiP,
	\end{romanenumerate}
\end{restatable}

\begin{proof}
Let $\mathcal I$ denote the given $\nicefrac{1}{\gamma}$-approximate $1$-IS, i.e., $|\mathcal {I}|\geq \nicefrac{\alpha_1}{\gamma}$.
First, we consider an instance of \ScMPropiP or \ScMPropiiP. Without loss of generality we assume $q=\lb{} + p +\rb{} = 1$. An optimal schedule distributes all jobs evenly over a maximum $1$-IS, i.e., we have
$
\OPT = \left\lceil\nicefrac{n}{\alpha_1}\right\rceil
$ as $q=1$.
Since each \qtime is $1$, we can find a schedule $\Pi$ with makespan 
$\|\Pi\| 
= \left\lceil \nicefrac{n}{|\I|} \right\rceil\cdot q 
\leq \left\lceil \nicefrac{\gamma\cdot n}{\alpha_1} \right\rceil 
\leq \lceil\gamma\rceil\cdot \left\lceil\nicefrac{n}{\alpha_1}\right\rceil 
= \lceil\gamma\rceil\cdot \OPT$ by distributing all jobs evenly on $\I$.

Next, we consider an instance of \ScMPropiiiP.
In order to construct $\Pi$, we use a list scheduling approach on $\I$. To this end, consider the jobs in an arbitrary order and iteratively assign the next job to a machine in $\I$ with minimum total completion time. As $\I$ is an independent set, $\Pi$ is feasible. Let $k$ be a job ending the latest in $\Pi$, i.e., $\|\Pi\|= C_k^\Pi=S_k^{\Pi}+q_k$.
By construction of $\Pi$,  all machines in $\I$ are processing some job at least up to time $S_k^{\Pi}$. Therefore,  $S_k^{\Pi}\leq \nicefrac{1}{|\I|}\sum_{j\in \J\setminus\{k\}}q_j
= (\nicefrac{1}{|\I|}\sum_{j\in \J}q_j)-\nicefrac{1}{|\I|}\cdot q_k$.

As in the proof of~\cref{cor:nonNested_polytime}, there exists an optimal schedule using only $\alpha_1$ machines. Therefore, the optimal makespan is lower bounded by the average load $\OPT\geq \frac{1}{\alpha_1}\sum_{j\in J}q_j$. Additionally, using the fact that 
$q_k\leq \OPT$ and $|\I|\leq m$, we obtain
\begin{align*}
\|\Pi\|&
=S_k^{\Pi}+q_k 
\leq
(\nicefrac{1}{|\I|}\sum_{j\in \J}q_j) +(1-\nicefrac{1}{|\I|})\cdot q_k 
\leq
\nicefrac{\alpha_1}{|\I |}\cdot \OPT +(1-\nicefrac{1}{|\I|})\cdot \OPT\\
&\leq (\gamma +1-\nicefrac{1}{m})\cdot \OPT.
\qedhere
\end{align*}
\end{proof}

The last approximation result generalizes the $(2-\nicefrac{1}{m})$-approximation by Graham~\cite{G66}.

\section{Identical Jobs}\label{sec:IdenticalNested}
In this section, we consider \ScMP with identical jobs, denoted by \IdenticalP. Complementary to the last section, we focus on the case of \emph{short blocking times}, i.e. $\max\{\lb{},\rb{}\}\leq p$. We denote an instance by \ScMIDSh with $n$ identical jobs with parameters $(\lb{},p,\rb{})$. By symmetry, we assume that $\lb{}\geq\rb{}$.

We start by introducing partial schedules which play an important role for \IdenticalP with short blocking times.

\subparagraph{$c$-Pattern.}
Consider a graph $G=(V,E)$  and a set of identical jobs with parameters $(\lb{},p,\rb{})$, where  $\lb{}\geq \rb{}$ and $q=\lb{}+p+\rb{}$. 
Let $c\in\mathbb{N}_{\geq 1}$ with $c\leq \lfloor\nicefrac{p}{\lb{}}\rfloor+1$ and let $\mathcal{I}=(I_1,I_2,\dots, I_c)$ be a  $c$-tuple of disjoint independent sets of $G$. A partial schedule of length $q+(c-1)\cdot \lb{}$ starting at some time $t$ is called a \emph{$c$-pattern on $\mathcal{I}$} if on each machine $i$ in $I_k$ with $k \in \{1,\dots,c\}$, there is one job starting at time $t+(k-1)\lb{}$. For an illustration consider \cref{fig:pattern}. 

%%%%%%%%%%%%%%%%%
\begin{figure}[htb]
		\centering
		\includegraphics[page=5,scale=1]{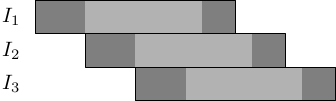}
		\caption{A $3$-pattern on three disjoint independent sets $I_1,I_2,I_3$.}
		\label{fig:pattern}
\end{figure}
%%%%%%%%%%%%%%

In order to show connections between conflict-free schedules and $c$-ISs, we extract a $c$-IS from a conflict-free schedule. 
We say a job $j$ \emph{blocks} a time $t$ in schedule $\Pi$ if one of its blocking times contains $t$, i.e., $t\in\left((S_j^{\Pi},S_j^{\Pi}+\lb{})\cup(C_j^{\Pi}-\rb{},C_j^{\Pi})\right)$. 
For a schedule~$\Pi$, we define the quantity 
\[
\beta_c^{\Pi}:=\max_{t_1< \ldots<t_c} \left\{\left|\bigcup_{k=1}^c \left\{i\in V: i\text{ processes a job in $\Pi$ which blocks time } t_k\right\}\right|\right\}.
\]
Observe that for each time $t$ the machines which process a job blocking time $t$ form an $1$-IS, because $\Pi$ is conflict-free. Hence, $\beta_c^{\Pi}$ corresponds to the cardinality of a $c$-IS, since machines are not counted more than once. 
For later reference, we note that $\beta_c^{\Pi}$ can be computed in polynomial time for a constant $c$.

\begin{restatable}{lemma}{computingAlpha}\label{lemma:computingAlpha}
	For a schedule $\Pi$ with $n$ jobs and a constant $c$, $\beta_c^{\Pi}$ can be computed in time polynomial in $n$.
\end{restatable}

\begin{proof}
	The schedule has $4n$ event times, namely the starting time of the blocking times and the processing time and its completion time.
	For some point in time $t$ between every two consecutive event points, we count the number of machines processing a blocking interval. By definition of $\beta_{c}^\Pi$, it suffices to check $\mathcal{O}(n^{c})$ tuples. 
\end{proof}

The definition of $\beta_{c}^{\Pi}$ helps us to bound the makespan of a schedule $\Pi$ from below. To this end, we first give an upper bound on the number of jobs starting within an interval of length $q$.

\begin{restatable}{lemma}{MaxNumJobsArbitraryInterval}\label{Lemma:Maximum_number_jobs_in_interval_arbitrary_schedule}
	For every schedule $\Pi$ of an instance \ScMIDSh with $p\geq \lb{}\geq \rb{}>0$ and every time $t\geq 0$, the number of jobs starting within the interval $I:=[t, t+q)$ is at most $\beta_{k +1}^{\Pi}$, where $k:=\lceil \nicefrac{p}{\lb{}} \rceil$.
\end{restatable}

\begin{proof}
	Because $I$ has length $q=(\lb{}+p+\rb{})$ and is half-open, at most one job starts on each machine within $I$.
	We partition the interval $I$ into $k+2$ disjoint (left-closed, right-open) intervals $I_1,\dots,I_{k+2}$, where the first interval has length~$\rb{}$,  the next $(k-1)$ intervals have length~$\lb{}$, interval $I_{k+1}$ has length $(p+\lb{})-(\rb{}+(k-1)\lb{})<\lb{}$ and the last interval $I_{k+2}$ has length $\lb{}$, see \cref{fig:maximum_starting_jobsAPP}.

	%%%%%%%%%%%%%%%%%
	\begin{figure}[htb]
		\centering
		\includegraphics[page=16]{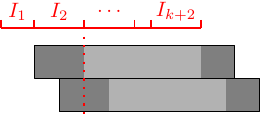}\qquad
		\includegraphics[page=17]{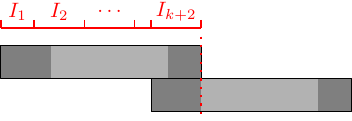}
		\caption{Illustration for the proof of \cref{Lemma:Maximum_number_jobs_in_interval_arbitrary_schedule}.}
		\label{fig:maximum_starting_jobsAPP}
	\end{figure}
	%%%%%%%%%%%%%%
	
	By $V_{\ell}$ we denote the set of machines that have a job starting in $I_{\ell}$.
	For all $\ell$ the jobs processed on machine $V_{\ell}$ block a point in time arbitrarily close to the right end of $I_\ell$.
	Additionally, every job processed on a machine in $V_1$ as well as every job processed on $ V_{k+2}$ block a point in time arbitrarily close to the right end of the whole interval $I$. Thus, the number of jobs can be bounded from above by  $\beta_{k+1 }^{\Pi}$. 
\end{proof}

Next, we show that an upper bound on the number of jobs starting within an interval of a specific lengths implies a lower bound on the makespan.

\begin{restatable}{lemma}{StartingJobsBound}\label{lem:startingJobs}
	
	Let $\Pi$ be a schedule for an instance \ScMIDSh with $\lb{}\geq \rb{}$ such that in every interval of length $L \leq q$ at most $\beta$ jobs start. Then, for \qtime $q:=\lb{}+p+\rb{}$, it holds that
	\[\|\Pi \|\geq L \lfloor\nicefrac{n}{\beta}\rfloor+
	\begin{cases}
	0& \text{if }\beta \text{ divides } n\\
	q& \text{otherwise.}
	\end{cases}
	\]
\end{restatable}
\begin{proof}
	We divide $\Pi$ into intervals of length $L$ starting with 0. By assumption, at most $\beta$ jobs start in each interval. Hence, the number of these intervals (where some job is processed) is at least~$\lfloor\nicefrac{n}{\beta}\rfloor$. Moreover, if $\beta$ does not divide $n$, then at least some job (of length $q$) starts after time $L\lfloor\nicefrac{n}{\beta}\rfloor+q$.
\end{proof}

\cref{Lemma:Maximum_number_jobs_in_interval_arbitrary_schedule,lem:startingJobs} immediately imply a lower bound on the makespan of any schedule.

\begin{restatable}{lemma}{LowerBoundIdentical}\label{corr:lower_bound}
	A schedule $\Pi$ of \ScMIDSh with $\lb{}\geq \rb{}>0$ has makespan   $\|\Pi\|\geq q\cdot\lceil\nicefrac{n}{\beta_{k+1}^{\Pi}}\rceil$, where $k:=\lceil \nicefrac{p}{\lb{}}\rceil$.
\end{restatable}

Using the fact that there exists no $O(m^{\varepsilon-1})$-approximation for computing a maximum $c$-IS~\cite{Haa99, LY93, Z06} and similar ideas as used for \cref{theorem:hardness_identical_long_blocking_time}, we can obtain an inapproximability result for the case of identical jobs for all fixed constants $\lb{},p,\rb{}$.

\begin{restatable}{theorem}{HardnessIdenticalShort}\label{theorem:hardness_identical_short_blocking_time}
	Unless $\P=\NP$, there exists no $\mathcal{O}(m^{1-\varepsilon})$-approximation for any $\epsilon > 0$ for \IdenticalP, even for any choice of fixed positive parameters $(\lb{},p,\rb{})$.
\end{restatable}
\begin{proof}
	By~\cref{theorem:hardness_identical_long_blocking_time} it suffices to restrict to the case in which the blocking times are short.
	Suppose for some $\kappa >0$ that there exists a $(\kappa m^{1-\varepsilon})$-approximation $\mathcal A$ for $\varepsilon > 0$.
	We assume $\lb{}\geq \rb{}$ and define the constant $k:=\lceil \nicefrac{p}{\lb{}}\rceil$. For an instance \ScMIDSh with $n\geq \alpha_{k+1}$, we consider the schedule $\Pi$ computed by $\mathcal A$.
	Because $\mathcal A$ is a $\kappa m^{1-\varepsilon}$-approximation and by \cref{corr:lower_bound},  the schedule $\Pi$ computed by $\mathcal A$ has makespan
	$\|\Pi\|
	\leq (\kappa m^{1-\varepsilon}) \cdot  OPT$
	and 
	$ 
	\|\Pi\|
	\geq q\cdot \left\lceil\nicefrac{n}{\beta_{k+1}^{\Pi}}\right\rceil\geq qn\cdot\nicefrac{1}{\beta_{k+1}^{\Pi}}.
	$
	Recall that, by \cref{lemma:computingAlpha}, we can compute a $(k+1)$-IS from $\Pi$ of size $\beta_{k+1}^\Pi$ in polynomial time.
	
	Moreover, we obtain a feasible schedule by repeatedly using $(k +1)$-patterns on a maximum $(k +1)$-IS, while leaving the other machines idle. 
	Recall that a $(k +1)$-pattern has length $(q+\lfloor \nicefrac{p}{\lb{}}\rfloor \lb{})\leq (q+k\lb{})$ and schedules $\alpha_{k+1}$ jobs.
	This yields the upper bound 
	$
	\OPT
	\leq (q+k\lb{})
	\cdot \left\lceil\nicefrac{n}{\alpha_{k+1}}\right\rceil 
	\leq 2(q+k\lb{})\cdot \nicefrac{n}{\alpha_{k+1}},
	$
	where the last inequality uses the fact $n\geq \alpha_{k+1}$. 
	All together we obtain
	\[
	\beta_{k+1}^{\Pi}\geq \frac{qn}{\|\Pi\|} \geq \frac{qn}{\kappa m^{1-\varepsilon}\cdot\OPT} \geq \frac{1}{\kappa m^{1-\varepsilon}}\cdot \frac{q}{2(q+k \lb{})}\cdot \alpha_{k+1},
	\]
	where $\nicefrac{q}{q+k \lb{}}$ is a constant $\leq 1$. By~\cref{lemma:computingAlpha}, $\mathcal{A}$ would therefore imply a $\mathcal{O}(m^{\varepsilon-1})$-approximation for computing a maximum $(k+1)$-IS; a contradiction~\cite{Haa99,LY93, Z06}.
\end{proof}

%---------------
 
 Note that this result holds for the case when the running time depends polynomially in $n$ (instead of $\log(n)$). However, if we are given a maximum $c$-IS of the conflict graph for some suitable $c$, we obtain approximation algorithms with performance guarantee better than $2$.

\begin{restatable}{theorem}{ApproxIdenticalShortOne}\label{thm:approximation_identical_short_max}
	If
	$0 < \rb{}\leq\lb{}\leq p$ and we are  given a maximum $\left(\lceil\nicefrac{p}{\lb{}}\rceil +1\right)$-IS of $G$, then
	\IdenticalP allows for  a $\left(1+\lfloor \nicefrac{p}{\lb{}}\rfloor \cdot \nicefrac{\lb{}}{q}\right)$-approximation, with  $\left(1+\lfloor \nicefrac{p}{\lb{}}\rfloor \cdot \nicefrac{\lb{}}{q}\right)<2$.
\end{restatable}
An analogous result for long blocking times is given in \cref{cor:nonNested_polytime}(\ref{cor:nonNested_polytimei}).
\begin{proof}

	Consider an instance \ScMIDSh and let $k = \lceil\nicefrac{p}{\lb{}}\rceil$.
	In order to construct a schedule~$\Pi$, we use $\left\lceil\nicefrac{n}{\alpha_{k+1}}\right\rceil$ many $(k+1)$-patterns on a maximum $(k+1)$-IS $\mathcal I$ of~$G$, while leaving all other machines idle. $\Pi$ yields a makespan of 
	$
	\|\Pi\|\leq(q+\lfloor \nicefrac{p}{\lb{}}\rfloor \lb{})\cdot \left\lceil\nicefrac{n}{\alpha_{k+1}}\right\rceil
	$.
	By \cref{corr:lower_bound}, it holds that 
	$\OPT\geq q\cdot \left\lceil\nicefrac{n}{\alpha_{k+1}}\right\rceil$.
	Moreover, note that $\lfloor \nicefrac{p}{\lb{}}\rfloor \cdot \nicefrac{\lb{}}{q}< 1$ and, hence, $\left(1+\lfloor \nicefrac{p}{\lb{}}\rfloor \cdot \nicefrac{\lb{}}{q}\right)<2$.
\end{proof}
Note that if the number of machines is constant, i.e. $m = O(1)$, the result above hold via complete enumeration for finding a maximum $c$-IS. 
Similarly, if we are given an approximate $c$-IS of the conflict graph for some suitable $c$, corresponding approximation results can be derived.

\begin{restatable}{theorem}{ApproxIdenticalShortTwo}\label{cor:approximation_identical_short_approxmax}
		If 
	$0 < \rb{}\leq\lb{}\leq p$ and we are given a $\nicefrac{1}{\gamma}$-approximate $\left(\lceil\nicefrac{p}{\lb{}}\rceil +1\right)$-IS of $G$, then \IdenticalP  allows for  a $2\gamma\left(1+\lfloor \nicefrac{p}{\lb{}}\rfloor \cdot \nicefrac{\lb{}}{q}\right)$-approximation, with \mbox{$2\gamma\left(1+\lfloor \nicefrac{p}{\lb{}}\rfloor \cdot \nicefrac{\lb{}}{q}\right) < 4\gamma$}.
\end{restatable}
An analogous result for long blocking times is given in \cref{theorem:ApprIndepSetLong}(\ref{item:approxi}).
\begin{proof} 
	Consider an instance \ScMIDSh and let $k = \lceil\nicefrac{p}{\lb{}}\rceil$. Moreover, let $\mathcal I$ be the given $(k+1)$-IS of size $\beta\geq \nicefrac{\alpha_{k+1}}{\gamma}$. We construct a schedule $\Pi$ by using $\left\lceil\nicefrac{n}{\beta}\right\rceil$ many $(k+1)$-patterns on $\mathcal I$, yielding 
	$
	\|\Pi\|\leq(q+\lfloor \nicefrac{p}{\lb{}}\rfloor \lb{})\cdot \left\lceil\nicefrac{n}{\beta}\right\rceil
	$. Moreover, \cref{corr:lower_bound} implies $\OPT \geq q\cdot \left\lceil\nicefrac{n}{\alpha_{k+1}}\right\rceil\geq q\cdot \nicefrac{n}{\alpha_{k+1}}$.
	If $n\leq \beta$, then also $n\leq \alpha_{k+1}$ and we obtain \[\frac{\|\Pi\|}{\OPT}\leq\frac{(q+\lfloor \nicefrac{p}{\lb{}}\rfloor \lb{})\cdot \left\lceil\nicefrac{n}{\beta}\right\rceil}{q\cdot \left\lceil\nicefrac{n}{\alpha_{k+1}}\right\rceil}=  \frac{(q+\lfloor \nicefrac{p}{\lb{}}\rfloor\lb{})}{q}<2.
	\] If $n > \beta$, we obtain $\left\lceil\nicefrac{n}{\beta}\right\rceil \leq  2\cdot\nicefrac{n}{\beta} \leq  2\gamma\cdot\nicefrac{n}{\alpha_{k+1}}$. Consequently, it holds
	\[
	\frac{\|\Pi\|}{\OPT}\leq \frac{(q+\lfloor \nicefrac{p}{\lb{}}\rfloor\lb{})\cdot 2\gamma \cdot \nicefrac{n}{\alpha_{k+1}}}{q\cdot \nicefrac{n}{\alpha_{k+1}}} \leq 2\gamma\cdot\left(1+\lfloor \nicefrac{p}{\lb{}}\rfloor \cdot \nicefrac{\lb{}}{q}\right) < 4\gamma. \qedhere
	\]
\end{proof}

\section{Unit Jobs}\label{sec:unit}
In this section, we consider the special case \UnitP  in which we are given $n$ identical unit jobs where $\lb{}=p=\rb{}=1$ for all jobs. 
By \cref{theorem:hardness_identical_short_blocking_time}, there exists no $\mathcal{O}(m^{1-\varepsilon})$-approximation for \UnitP on general graphs, unless $P=NP$.
 In the following, we consider three graph classes for which \UnitP can be solved in polynomial time. Specifically, we study complete graphs, stars, and bipartite graphs. 
 
 At several places, we use the fact that  when all parameters are integers, there exists an optimal schedule with integral starting times.
 
 \begin{restatable}{lemma}{integralSchedules}\label{lem:integral}
 	If $\lb{j}$, $p_j$, $\rb{j}$ are integral for all $ j\in \J$, every feasible schedule $\Pi$ can be transformed into a feasible schedule $\Pi^*$ such that the starting time of each job is intergral and the makespan does not increase, i.e.,  $\|\Pi^*\|\leq \|\Pi\|$. 
 \end{restatable}
 
 \begin{proof}
 	Let $S_j^{\Pi}$ denote the starting time for each job $j\in J$. We define $\Pi^*$ by  $S_j^{\Pi^*}:=\lfloor S_j^{\Pi}\rfloor$ for each job $j$. It remains to show that $\Pi^*$ is feasible. Let $j$ and $j'$ be two jobs such that two of their blocking times $b$ and $b'$ intersect in ${\Pi^*}$. By integrality of $\Pi^*$ and the blocking times, they intersect in at least one time unit. 
 	
 	Without loss of generality, we assume that $b'$ does not start before $b$ in $\Pi$. With slight abuse of notation, $S_{b}^\Pi$  denotes the starting time of the blocking time $b$ in $\Pi$. Then, $S_{b'}^\Pi-S_{b}^\Pi$ and $S_{b'}^{\Pi^*}-S_{b}^{\Pi^*}$ differ by strictly less than 1. Hence, if they intersect in at least 1 time unit in $\Pi^*$, then they also intersect in $\Pi$. Therefore, by feasibility of $\Pi$, $j$ and $j'$ are scheduled on conflict-free machines. 
 \end{proof}
 
 \subsection{Complete graphs}
 
 Complete graphs play a special role in the context when machines share a single resource, e.g., a resource  for pre- and post-processing which can only be accessed by a single machine at once.

  \begin{restatable}{lemma}{completeGraph}\label{lem:UNITcomplete}
  	For every $n$, an optimal schedule for \UnitP{$(K_m,n)$} can be computed in time linear in $\log n$. In particular, for $m\geq 2$, it coincides with an optimal schedule for $K_2$ of makespan $4\lfloor\nicefrac{n}{2}\rfloor+3(n \mod 2)$. 
  \end{restatable}
\begin{proof}
	Let $\Pi$ be an optimal schedule of \UnitP{$(K_m,n)$}. Note that for each point in time~$t$, the set of machines processing a job at time $t$ induces a $K_1$ or $K_2$. Moreover, all vertices in $K_m$, $m\geq 2$, play the same role and hence we may shift all jobs to the same two vertices. Thus, we may reduce our attention to \UnitP{$(K_2,n)$}. It is easy to check that two jobs are optimally processed in time 4. This implies the claim.
\end{proof}

\subsection{Bipartite Graphs}\label{sec:complete_bipartite}
In this subsection, we show how to compute optimal schedules on bipartite graphs in polynomial time.

 Studying bipartite graphs in this setting is interesting due to several properties: Firstly, we can compute 1,2-IS of bipartite graphs in polynomial time. Consequently, there is an indication that the problem is easier on bipartite graphs, e.g.,  \cref{thm:approximation_identical_short_max} yields a $\nicefrac{4}{3}$-approximation for bipartite graphs.
Secondly, for every point in time, the set of active machines induces a bipartite graph (\cref{prop:bip}). Hence, understanding optimal schedules of bipartite graphs also yield  \emph{local optimality criteria} of schedules for all graphs.

  \begin{restatable}{observation}{BipSubgraphs}\label{prop:bip}
  	Consider an instance \UnitP on a graph $G$ and a feasible schedule~$\Pi$.
  	For every point in time $t$, the set of machines processing a job at time $t$ in $\Pi$ induces a bipartite subgraph of $G$.  
  \end{restatable}

We derive our polynomial time algorithm for bipartite graphs in three steps.
Firstly, we derive a structural property of optimal schedules on stars (\cref{sec:stars}).
Then, we exploit these structural insights for general bipartite graphs by considering a subgraph whose components are stars (\cref{sec:general_bipartite}). 
The special structure is a slightly relaxed notion of 1- and 2-patterns.\medskip

\noindent \textbf{A/B-patterns.}
An \emph{A-pattern} on a graph $H$ is a 
 a $1$-pattern on some maximum $1$-IS of $H$.
 A \emph{B-pattern} on $H$ is a $2$-pattern on some maximum $2$-IS of $H$. Note that the difference to 1- and 2-patterns is  that we do not specify the 1- and 2-ISs. 

An \emph{AB-schedule} on $H$ consists of  A- and/or B-patterns.
Let $n\in\mathbb N$. We say \UnitP{$(H,n)$} admits an optimal AB-schedule if there exists an optimal schedule that can be transformed into an AB-schedule on $H$ by possibly adding more jobs without increasing the makespan.

\subsubsection{Stars}\label{sec:stars}
In this subsection, we describe a polynomial time algorithm to find optimal schedules on stars. 
A \emph{star} is a complete bipartite graph $S_{\ell}:=K_{1,\ell}$ on $\ell+1\geq 2$ vertices. For $\ell\geq 2$, $S_{\ell}$ has $\ell$ leaves and a unique \emph{center} of degree $\ell$. For $\ell=1$, either vertex can be seen as the center of~$S_{1}$.

\begin{restatable}{lemma}{BipartitePatterns}\label{lemma:Bipartite_Patterns}
	For every star $S$ and every $n$, \UnitP{$(S,n)$} admits an optimal AB-schedule.
\end{restatable}

\begin{proof}
	For simplicity of the presentation, we first consider a star $S_\ell$ with $\ell\geq 2$. We show that there exists an optimal schedule of \UnitP{$(S,n)$} to which we can possibly add jobs (without increasing the makespan) such that all $\ell$ leaves have the same induced schedule. To this end, we can consider an optimal schedule and determine a leaf processing a maximum number of jobs. Changing all leaves to this schedule may only increase the number of processed jobs and yields a valid schedule.
	
	Similarly, the above property holds trivially for $S_1$ by replacing leaf with non-center. 
	
	Finally, for each job $j$ processed on a leave (non-center) of $S_\ell$, there exist two cases: If no job is processed on the center during the processing time of $j$, we obtain an A-pattern. Otherwise, we obtain a B-pattern.
 $\mathcal{M}_3=\emptyset$, then no job is running/processing at $t+4$.
\end{proof}

\begin{corollary}\label{cor:stars}
	For every star $S$  and every $n$, an optimal schedule for \UnitP{$(S,n)$} can be computed in time linear in $\log n$ and $|S|$.
	
	Specifically, for every $S$, there exists a $X\in\{A,B\}$ such that an optimal schedule has at most 2 $X$-patterns, i.e., an optimal schedule has makespan
	\[
	\min_{k=0,1,2} \left\{4\bigg\lceil\frac{n-k\ell}{\ell+1}\bigg\rceil+3k, 3\bigg\lceil\frac{n-k(\ell+1)}{\ell}\bigg\rceil+4k \right\},% \text{ where } \lceil\cdot \rceil: \mathbb R\to \mathbb N, x\mapsto ?,
	\]
	where $\lceil\cdot \rceil$ denotes the usual ceiling function; however, for negative reals it evaluates to 0.
\end{corollary}

\begin{proof}

	By \cref{lemma:Bipartite_Patterns}, there exists an optimal AB-schedule. 
	For a star  $S_\ell$, an A-pattern processes $\ell$ jobs in time 3 and a B-pattern processes $(\ell+1)$ jobs in time 4.
	An AB-schedule for at least $n$ jobs with exactly $k$ A-patterns  has a makespan of $4\lceil\frac{n-k\ell}{\ell+1}\rceil+3k$.
	Similarly, an AB-schedule for at least $n$ jobs with exactly $k$ B-patterns  has a makespan of $3\lceil\frac{n-k(\ell+1)}{\ell}\rceil+4k$.
 	 
 	  For $\ell\leq 2$, an optimal schedule contains  at most 2 A-patterns. While 3 A-patterns process $3\ell$ jobs in time 9, 2 B-pattern process $2(\ell+1)\geq 3\ell$ jobs in time 8. Thus, any 3 A-patterns can be replaced by 2 B-patterns.

	For $\ell\geq 3$, an optimal schedule contains at most 2 B-patterns: While  3 B-patterns process $3(\ell+1)$ jobs in time 12, 4 A-patterns finish $4\ell\geq 3(\ell+1)$ jobs in time 12. Thus, any 3 B-patterns can be replaced by 4 A patterns.
	
Therefore, for each star we only have to compare at most three different schedules with $k=0,1,2$ A- or B-patterns, respectively. Taking one with minimum makespan induces an optimal solution.
\end{proof}

We remark that for $S_{3}$, both 4 A-patterns and 3 B-patterns finish  twelve jobs in time~12. We exploit this fact later.

\subsubsection{The approach for general bipartite graphs}\label{sec:general_bipartite}
In the last section, we saw that we can restrict our attention to AB-schedules for stars. This property does not generalize to all bipartite graphs. In fact, it does not even hold for trees as illustrated in~\cref{fig:DanielsBeispiel}.
\begin{observation}
	There exists a tree $T$ such that no optimal schedule for \UnitP{}$(T,n)$ is an AB-schedule with respect to $T$.
\end{observation}

\begin{figure}[htb]
	\centering
			\includegraphics[page=19]{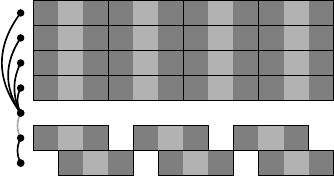}
\caption{Optimal schedule on a tree with $7$ machines and $22$ jobs with makespan $12$.}
\label{fig:DanielsBeispiel}
\end{figure}

Note, however, that the optimal schedule shown in \cref{fig:DanielsBeispiel} is comprised 
 of two $AB$-schedules on the two stars obtained by deleting the gray edge. 
 Inspired by this observation, we show that this property generalizes to all bipartite graphs. This is a crucial step for our polynomial time algorithm.

 \begin{theorem}\label{theorem:bipartite}
 	For every bipartite graph $G$ and every $n$, an optimal schedule for
 	\Unit can be computed in polynomial time.
 \end{theorem}

	We solve the problem on each connected component. If a component contains one vertex, we use \cref{lem:UNITcomplete}. Thus, it remains to show how to compute optimal schedules for connected bipartite graphs on at least two vertices.
	
	The key idea is to find a spanning subgraph $H$ of $G$ for which AB-schedules with respect to $H$ are among the optimal schedules for $G$. Given such a subgraph, we can then exploit the special structure of AB-schedules to compute an optimal schedule in polynomial time.
	In particular, we are interested in subgraphs that consists of stars and allow so called I-, II-, and III-colorings.

\subparagraph{Star forests and I,II,III-colorings.}
	Let $G$ be a (bipartite) graph. A subgraph $H$ of $G$ is a \emph{star forest of $G$} if each component of $H$ is a star and $H$ contains all vertices of $G$. Whenever we talk about a star of $H$, then we refer to  an entire component of $H$.

Moreover, we make use of particular vertex subsets, denoted by $A_i$ and $B_i$, to process the jobs. The idea  is to schedule A-patterns on the $A_i$'s and $B$-patterns on the $B_i$'s. The properties ensure that the resulting schedules are valid for $H$ and $G$.

	Let $H$ be a star forest of a bipartite graph $G$. A vertex subset $A_1$ is a \emph{I-coloring} of~$(G,H)$ if it is a maximum independent set of both $G$ and $H$. Clearly, this is equivalent to the following conditions: no two vertices of $A_1$ are adjacent in $G$ and for each star $S$ of $H$,  $A_1\cap S$ is a maximum independent set of $S$. 
	 A I-coloring allows to schedule an A-pattern on $H$ (by placing one job on each machines in $A_1$), yielding a valid schedule for $G$, see \cref{fig:ColoringExamples}(left).
	 
	 \begin{figure}[hb]
	 	\centering
	 	\includegraphics[page=10]{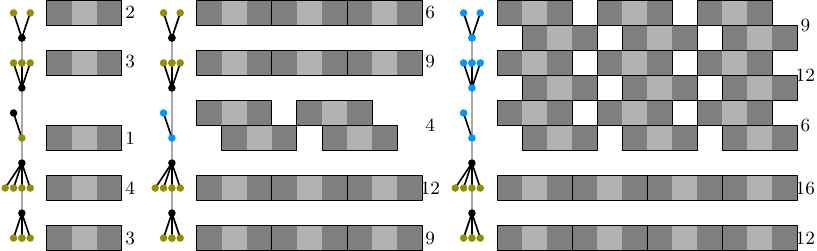}
	 	\caption{A graph and  a star forest with a I-, II-, and III-coloring and the corresponding schedules. Vertex sets $A_i$ are colored in asparagus and $B_i$ in blue.}
	 	\label{fig:ColoringExamples}
	 \end{figure}

Two disjoint vertex subsets $A_2,B_2$ are a \emph{II-coloring} of $(G,H)$, if no vertex of $A_2$ is adjacent to another vertex from $A_2\cup B_2$ in $G$ and the following properties hold: (i) for each $S=S_{\ell}$, $\ell\geq 3$, $A_2\cap S$ is a maximum independent set of $S$,  (ii) $B_2$ contains the vertices of  each $S_{1}$, and (iii) for each $S_{2}$, either $A_2\cap S_2$ is a maximum independent set of $S_{2}$ or $B_2$ contains the vertices of $S_{2}$.
Observe that a II-coloring allows us to schedule 3 A-patterns on stars with leaves in $A_2$ and 2 B-patterns on stars with vertices  in $B_2$,  see \cref{fig:ColoringExamples}(middle).

 Two disjoint vertex subsets $A_3,B_3$ are a \emph{III-coloring} of $(G,H)$, if no vertex of $A_3$ is adjacent  to another vertex from $A_3\cup B_3$ in $G$ and the following properties hold: (i) for each $S=S_{\ell}$, $\ell\geq 4$, $A_3\cap S$ is a maximum independent set of $S$,  (ii) $B_3$ contains the vertices of  each $S_{1}$ and each $S_{2}$, and (iii) for each $S_{3}$,  $A_3\cap S_{3}$ is a maximum independent set of $S_{3}$ or $B_3$ contains the vertices of $S_{3}$.
A III-coloring allows us to schedule 4 A-patterns on stars with leaves in $A_{3}$ and 3 B-patterns on stars with vertices in $B_{3}$,  see \cref{fig:ColoringExamples}(right).
\medskip

The proof of \cref{theorem:bipartite} consists of two parts: Firstly, we show how to compute optimal schedules when we are given a star forest of $G$ with a I-, II-, and III-coloring (\cref{lem:givenForest}). Secondly, we present an algorithm to compute  a star forest  and the colorings (\cref{lem:algo_correctpoly}).

\subsubsection{Optimal schedules for a given star forest with I,II,III-colorings}
The key idea is to find a star forest of the graph such that the optimal AB-schedule on each star yields a feasible schedule with respect to $G$, hence is also optimal.

	\begin{lemma}\label{lem:givenForest}
		Let $H$ be a star forest of a connected bipartite graph $G$ on at least two vertices.  Given a I-coloring $A_1$, a  II-coloring $(A_2,B_2)$ and a  III-coloring $(A_3,B_3)$, there exists a polynomial time algorithm to compute an optimal schedule for \Unit. 
	\end{lemma}

	\begin{proof}
		Let $\Pi'$ be an optimal schedule for \Unit. By \cref{lemma:Bipartite_Patterns}, there exists an optimal AB-schedule~$\Pi$ for \UnitP{$(H,n)$}.
		Because $H$ is a subgraph of $G$, we have $\|\Pi\|\leq \|\Pi'\|$.
		First, we show how to solve \Unit for small optimal makespan values.
		
		\begin{claim}\label{claim:shortAB}
			If $\|\Pi\|\leq 20$, then there exists an optimal AB-schedule $\Pi^*$ on $H$ that is feasible for $G$ (according to \cref{table:patterns}).
		\end{claim}

		Note that  $\|\Pi\|\geq q =3$ and that there is no  AB-schedule with $\|\Pi\|=5$.
		We first show how to define a  schedule $\Pi^*$ of makespan $\|\Pi\|$ according to \cref{table:patterns} that has as many jobs as $\Pi$. Afterwards, we show that $\Pi^*$ is feasible with respect to $G$.
		To this end, we first concentrate on the four types of components in $H$. 
		For $C\in\{S_1,S_2,S_3,S_{\ell\geq 4}\}$, let $L_{C}$ denote the set of all makespan (of at most 20) from schedules obtained by \cref{cor:stars}.
		For $C\in\{S_1,S_2,S_3,S_{\ell\geq 4}\}$ and $\|\Pi\|\in [20]\setminus \{1,2,5\}$, we use the optimal AB-schedule with makespan $\|\Pi\|$ (or the maximum value in $L_C$ that is $\leq\|\Pi\|$). 
		However, we modify some schedules in order to guarantee feasibility later: modified entries are marked by an asterisk in~\cref{table:patterns}. 
		
		For $S_1$ and for $\|\Pi\| \equiv 2 \pmod{4}$, we use 2 A-patterns and $(\lfloor\nicefrac{\|\Pi\|}{4}\rfloor -1)$ B-patterns instead of $\lfloor\nicefrac{\|\Pi\|}{4}\rfloor$ B-patterns. Since 2 A-patterns and 1 B-pattern differ in length 2, the modified schedule finishes within $\|\Pi\|$. It also schedules at least as many jobs as $\Pi$, because 1 A-pattern contains one job while 1 B-pattern contains two jobs.
		
		For $S_2$ and for $\|\Pi\| \equiv 1 \pmod{4}$, we also allow to schedule 3 A-patterns and $\lfloor\nicefrac{\|\Pi\|}{4}\rfloor-2$ B-patterns besides the optimal schedule using $\lfloor\nicefrac{\|\Pi\|}{4}\rfloor$ B-patterns.  As 3 A-patterns and 2 B-patterns differ in length 1, the modified schedule finishes within $\|\Pi\|$. Moreover, both 3 A-patterns and 2 B-patterns contain six jobs.
		\cref{table:patterns} displays the resulting patterns on the components. The schedule $\Pi^*$ is constructed as follows: Each A-pattern is scheduled on $A_1$,  each B-pattern on $V$, 3A-patterns on $A_2$, 2B-patterns on $B_2$, 4A-patterns on $A_3$, and 3B-patterns on $B_3$.
		\begin{table}[htb]
			\centering
			\caption{AB-schedules on the stars of $H$. The number before A and B indicates the number of A- and B-patterns. Consecutive patterns are separated by commas.}
			\label{table:patterns}
			\begin{tabular}{  c | c | c | c | c  }
				$\|\Pi\|$ & $S_{1}$ & $S_{2}$ & $S_{3}$ & $S_{\ell}$, $\ell\geq 4$   \\ \hline
				\rowcolor[gray]{0.9}1 & - & - & - & -  \\
				2 & - & - & - & -  \\ 
				\rowcolor[gray]{0.9}3 & A & A & A & A \\
				4 & B & B & B & B \\
				\rowcolor[gray]{0.9}5 & - & - & - & - \\
				6 & 2A* & 2A & 2A & 2A \\
				\rowcolor[gray]{0.9}7 & A,B & A,B & A,B & A,B \\
				8 & 2B & 2B & 2B & 2B \\
				\rowcolor[gray]{0.9}9 & 2B & 3A or 2B* & 3A & 3A \\
				10 & 2A,B* & 2A,B & 2A,B & 2A,B \\
				\rowcolor[gray]{0.9}11 & A,2B & A,2B & A,2B & A,2B \\
				12 & 3B & 3B & 4A or 3B & 4A \\
				\rowcolor[gray]{0.9}13 & B,2B & B,(3A or 2B)* & B,3A & B,3A \\
				14 & 2A,2B* & 2A,2B & 2A,2B & 2A,2B \\
				\rowcolor[gray]{0.9}15 & A,3B & A,3B & A,(4A or 3B) & A,4A \\
				16 & B,3B & B,3B & B,(4A or 3B) & B,4A \\
				\rowcolor[gray]{0.9}17 & 2B,2B & 2B,(3A or 2B)* & 2B,3A& 2B,3A \\
				18 & 2A,3B* & 2A,3B & 2A,(4A or 3B) & 2A,4A \\
				\rowcolor[gray]{0.9}19 & A,B,3B & A,B,3B & A,B,(4A or 3B) & A,B,4A \\
				20 & 2B,3B & 2B,3B & 2B,(4A or 3B) & 2B,4A  
			\end{tabular}	
		\end{table}
		It remains to show that scheduling according to \cref{table:patterns} is feasible with respect to~$G$.
		
		By definition of $A_1$, scheduling an A-pattern on $A_1$ yields a feasible schedule for \Unit. This is used for the A-patterns in the cases where $\|\Pi\|\in \{3,6,10,11,14,15,18,19\}$.
		
		Because $G$ is bipartite, a B-pattern can be scheduled on $V$. This is used for the B-patterns in the cases where $\|\Pi\|\in\{4,7,8,10,11,13,14,16,17,19,20\}$.
		
		Because $(A_2,B_2)$ is a II-coloring, scheduling 3 A-patterns on $A_2$ and 2 B-patterns on $B_2$ is feasible. This is used for the case where $\|\Pi\|\in\{9,13,17\}$.
		
		Because $(A_3,B_3)$ is a III-coloring, scheduling 4 A-patterns on $A_3$ and 3 B-patterns on~$B_3$ is feasible. This is used for the case where $\|\Pi\|\in\{12,15,16,18,19,20\}$. This proves \cref{claim:shortAB}.
		
		\begin{claim}\label{claim:longAB}
			There exists an optimal schedule that is comprised of blocks of length 12 following row 12 of \cref{table:patterns} and one rest block of length at most 20 following \cref{table:patterns}. 
		\end{claim}

		By \cref{claim:shortAB}, we need to show Claim \cref{claim:longAB} for $\|\Pi\|\geq 21$.
		We may assume that for any two different components of $H$, their makespans in $\Pi$ differ by at most 3. Suppose there exists components $C_1$ and $C_2$ such that the makespan of $C_1$ exceeds the makespan of $C_2$ by at least 4. Deleting a last job from $C_1$ and inserting it after the last job of $C_2$, the makespan difference decreases.
		As a consequence, the makespan in $\Pi$ on each component is at least~18. This implies that  $\Pi$ schedules at least 4 A-patterns or 3 B-patterns both of length 12 on each connected component: Let $C$ be a component of $H$. If $\Pi$ has at most 3 A-patterns and at most 2 B-patterns on $C$, then the makespan on $C$ is at most $3\cdot 3+2\cdot 4=17$. A contradiction. Therefore, if $\|\Pi\|\geq 21$, we modify $\Pi$ by scheduling the 4 A-patterns on $A_3$ and the 3 B-patterns on $B_3$. By definition of a III-coloring of $(A_3,B_3)$ this yields a feasible subschedule with respect to $G$. We repeat this procedure until the makespan  of the remaining patterns is at most 20. This proves \cref{claim:longAB}. \medskip
			
		We can compute the schedule obtained by \cref{table:patterns} for each possible value $r\in [20]\setminus \{1,2,5\}$ on each connected component using the given I-,II- and III-colorings and filling it up with blocks of 12 following row 12 of \cref{table:patterns}. By \cref{claim:longAB}, the schedule with minimum makespan is an optimal schedule.
\end{proof}

This allows us to compute an optimal schedule for \Unit in polynomial time when a star forest together with its colorings are given.

%%%%%%%%%%%%%%%%%%%%%%%%%%%%%%%

\subsubsection{Computing a star forest with the colorings}
In this section, we present a polynomial time algorithm to find a star forest of a connected bipartite graph that admits I,II,III-colorings as used in \cref{lem:givenForest}.

We compute such a star forest in three steps.
First, we find a star forest admitting a I-coloring.
To ensure the existence of a II- and III-coloring, some adjustments may be necessary. Specifically, we will modify $H$ by switching edges along alternating paths.

\subparagraph{Alternating paths.} Let $H=(V,E')$ be a star forest of a bipartite graph $G=(V,E)$. 
Let $C_1,\ldots,C_k$ be distinct stars of $H$ and $P$ be  a path in $G$ on the vertices $v_1,v_2\dots, v_{2k-2}$ with the following properties:
\begin{itemize}
	\item for even $i$, $v_i$ is a leaf of star $C_{\nicefrac{i}{2}+1}$
	\item for odd $i$, $v_i$ is the center of star $C_{ \nicefrac{(i+1)}{2}}$	
	\item  edge $\{v_i,v_{i+1}\} \in E'$ if and only if $i$ is odd.
\end{itemize}
The path $P$ is an \emph{alternating path of type II} if $C_1\simeq S_{1}$ , $C_i\simeq S_{2}$ for all $i=2,\dots,k-1$ and $C_k\simeq S_{\ell}$ with $\ell\geq 3$. For an illustration of alternating paths of type {II}, see \cref{fig:alter_path_type2}.

\begin{figure}[htb]
	\centering
	\includegraphics[page=5]{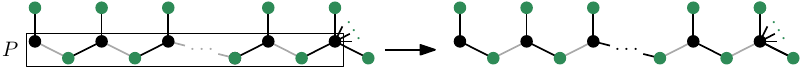}
	\caption{Alternating path $P$ of type II.  Black edges belong to $H$, gray edges to $G\setminus H$.
	}
	\label{fig:alter_path_type2}
\end{figure}

We say $P$ is an \emph{alternating path of type III} if $C_1\simeq S_{2}$ , $C_i\simeq S_{3}$ for all $i=2,\dots,k-1$ and $C_k\simeq S_{\ell}$ with $\ell\geq 4$. For an illustration of an alternating path of type {III}, see \cref{fig:alter_path_type3}.

\begin{figure}[!htb]
	\centering
	\includegraphics[page=6]{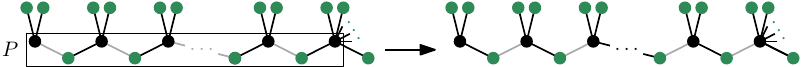}
	\caption{Alternating path of type III. Black edges belong to $H$, gray edges to $G\setminus H$.}
	\label{fig:alter_path_type3}
\end{figure}

If a star forest contains either type of alternating path, we show how to modify it such that the resulting star forest does not include any alternating path of type II or III. Here, we make use of the following observation.
\begin{observation}\label{observation:augmenting}
Let $H=(V,E')$ be a star forest of $G$ containing an alternating path $P$ of type II or III. Then,  $H':=(V,E'\Delta P)$ is a star forest with the same set of leaves as $H$, where $\Delta$ denotes the symmetric difference.
\end{observation}
Furthermore, we use the following lemma.
\begin{lemma}\label{lem:augmenting3to2}
Let $H=(V,E')$ be a star forest of $G$ without any alternating paths of type~II. Let $P$ be an alternating path $P$ of type III. Then the star forest $H'=(V,E'\Delta P)$ contains no alternating paths of type II.
\end{lemma}
\begin{proof}
Let $C_{1}, \dots C_{k}$ be the stars of the alternating path $P$ in $H$ and let $C'_{1}, \dots C'_{k}$ denote the corresponding stars in $H'$. Observe that $C'_{i}\simeq C_i\simeq S_3$ for all $i\in\{2,\dots,k-1\}$. Moreover, $C'_1=S_3$ and $C'_k=S_\ell$ for some $\ell \geq 3$.
For the purpose of a contradiction, suppose that $H'$ contains 
an alternating path $P_2$ of type II. 
Clearly,  $P_2$ and $P$ intersect; otherwise $P_2$ is also contained in $H$. Specifically, $P_2$ ends in some~$C'_i$ with $i\in\{1,\dots,k\}$. If $i>1$, then $P_2$ and $P$ share exactly the center of $C'_i$ and thus $H$ contains $P_{2}$ as well.  A contradiction. 
If $P_2$ ends in $C'_1$, then $H$ contains an alternating path of type II ending in $C_2$ that goes via  $C_1\simeq S_2$. Again, a contradiction.
\end{proof}

\begin{algorithm}[htbp]
  \caption{Computing a star forest and  I,II,III-colorings..}
  \label{alg:decompose}
  \begin{algorithmic}[1]
  \State Input: Connected bipartite graph $G=(V,E)$ with $|V|\geq 2$.
  \State Output: Star forest $H$ and I,II,III-colorings $A_1$,  $(A_2,B_2)$,  $(A_3,B_3)$.
    \phase{Initial star forest}
      \State Compute a maximum matching $M$ and a maximum independent set $I$ of $G$
      \State Set $U:=V\setminus I$ (vertex cover).
      \State Set $E': = M$ and $V' := V\setminus\bigcup_{e \in M} e$.
      \While{$\exists$ $v \in V'$}
      \State Find $u\in U$ such that  $\{u,v\} \in E$. 
      \State 
       Add $\{u,v\}$ to $E'$ and delete $v$ from $V'$.
      \EndWhile
    \phase{Removing alternating paths of type II}
    \While{$\exists$ alternating path $P$ of type II}
    \State $E' = E' \Delta P$
    \EndWhile
    \phase{Removing alternating paths of type III}
    \While{$\exists$ alternating path $P$ of type III}
    \State $E' = E' \Delta P$
    \EndWhile
    \phase{Computing the colorings}
    \State $H:=(V,E')$
    \State $A_1:=I$
    \State For each $S_\ell$ in $H$ with $\ell\geq 3$, add leaves of $S_\ell$ to $A_2$. 
    \State For each $S_1$ in $H$, add vertices of $S_1$ to $B_2$.
    \While{$\exists$  $S_2$ in $H$ such that its center  is adjacent to a vertex of $A_2$ in $G$}
   	  \State Add leaves of $S_2$ to $A_2$. 
    \EndWhile
    \State For each $S_2$ in $H$ with $V(S_2)\cap A_2=\emptyset$, add  vertices of  $S_2$ to $B_2$.
    \State For each $S_\ell$ in $H$ with $\ell\geq 4$, add leaves of $S_\ell$ to $A_3$.
    \State For each $S_\ell$  in $H$ with $\ell\in\{1,2\}$, add vertices of $S_\ell$ to $B_3$.
    \While{$\exists$ some $S_3$ in $H$ such that center $v$ of $S_3$ is adjacent to some $w\in A_3$ in $G$}
    \State Add leaves of $S_3$ to $A_3$.
    \EndWhile
    \State For each $S_3$ in $H$  with $V(S_3)\cap A_3=\emptyset$, add vertices of all $S_3$ to $B_3$.
    \State \Return  $ H,A_1,(A_2,B_2),(A_3,B_3)$ 
  \end{algorithmic}
\end{algorithm}

 Algorithm~\ref{alg:decompose} computes a star forest and I,II,III-colorings in polynomial time.

\begin{lemma}\label{lem:algo_correctpoly}
Algorithm~\ref{alg:decompose} returns a star forest $H$ of $G$ and  I-,II- and III- colorings $A_1,(A_2,B_2)$ and $(A_3,B_3)$ of $(G,H)$, respectively,  in time polynomial in~$G$.
\end{lemma}
\begin{proof}
Let $M$ be a maximum matching and $I$ be a maximum independent set of $G$. Both can be found in polynomial time using the maximum flow algorithm to find a maximum matching in bipartite graphs  \cite[Theorem 10.5]{KV18}. Observe that the complement of a maximum independent set is a minimum vertex cover $U:=V\setminus I$.
By K\H{o}nig's Theorem~\cite{K31}, it holds that $|U|=|M|$.  In particular, every edge in $M$ contains exactly one vertex of $U$ and every vertex in $V' := V\setminus\bigcup_{e \in M} e$ is not contained in $U$. Therefore, for every $v\in V'$, there exists $u\in U$ such that $\{u,v\}\in E$, i.e., line 7 in Phase 1 is well-defined. 
It remains to show that $(V,E')$ is a star forest at the end of Phase 1. To this end, note that every edge of $E'$ is incident to exactly one vertex $u\in U$. 
Thus, every vertex in $U$ is the center of a star (on at least 2 vertices).
Moreover, by \cref{observation:augmenting}, modifying the star forest $(V,E')$  along an alternating path in Phase 2 and 3 results again in a star forest. Thus, the algorithm returns a star forest (if it terminates).

For the runtime, it is important to observe that no new $S_1$ is created in Phase 2 and no new  $S_2$ is created in Phase 3. Therefore, the number of iterations in Phase 2 is bounded by the number of $S_1$'s after Phase 1 and in Phase 3 by the number of $S_2$'s after Phase 2. An alternating path of type II and III can also be found in polynomial time by fixing a $S_{1}$ or $S_{2}$, respectively, and use breadth first search. Because Phase 4 clearly runs in polynomial time as well, \cref{alg:decompose} terminates in polynomial time.

By~\cref{lem:augmenting3to2} and the fact that no alternating path of type II exists after Phase 2, after the end of Phase 3, there exists neither an alternating path of type II nor one of type III. We will exploit this fact now to prove correctness of the colorings.

By definition, $A_1:=I$ is a maximum independent set of $G$. We argue that $A_1$ is also an maximum 1-IS of $H$. Note that while $U$ constitutes the centers of the stars, $V/U=I=A_1$ are the leaves of the stars after Phase 1. Hence the claim is true after phase 1.
 \cref{observation:augmenting} ensures that this property is maintained in Phase 2 and 3.
 Hence, $A_1$ is a I-coloring.
 
For the type II-coloring, the algorithm ensures that $A_2$ contains the leaves of all $S_{\ell}$  with $\ell\geq 3$ and that $B_2$ contains  all vertices of $S_1$. Hence, we only need to pay attention to $S_2$'s. The algorithm inserts the leaves of stars $S_{2}$ into $A_2$ if their center is adjacent to a vertex in $A_2$, as long as it is possible. The vertices of all remaining $S_2$'s are inserted into $B_2$.
Thus, the properties (i), (ii) and (iii) of a II-coloring are fulfilled by $(A_2,B_2)$. It only remains to show that no vertex of $A_2$ is adjacent to another vertex of $A_2\cup B_2$ in $G$. Because $A_2\subset I$, no two vertices in $A_2$ are adjacent in $G$. Suppose there is a vertex $a\in A_2$ adjacent in $G$ to a vertex $b\in B_2$. 
Let $S^a$ and $S^b$ be the star containing $a$ and $b$, respectively.
By construction $a$ is a leave of $S^a$. Moreover,  $b\in U$; otherwise $a,b\in A_1$, a contradiction to the fact that  $A_1$ is an independent set.
If $S^b\simeq S_2$, then the center $b$ is adjacent to vertex $a\in A_2$,  and the algorithm ensures that the leaves of $S^b$ are contained in $A_2$, and hence, $b\notin B_2$. Thus, $S^b\simeq S_1$. 

If $S^a\simeq S_\ell,\ell\geq 3$, then there exists an alternating path of type II starting in $b$ and ending in the center of $S^a$, see \cref{fig:alterPath} (left). A contradiction. 

If $S^a\simeq S_2$, the fact $a\in A_2$ implies that there exists a star $S_\ell$, $\ell\geq 2$, with a leaf  adjacent to the center  of $S^a$. If $\ell=2$, there exists a further star whose leaf is adjacent to the considered star, see \cref{fig:alterPath} (right). Repeating this argument, the containment of $a\in A_2$ can be traced back to a star $S_\ell,\ell\geq 3$, and yields an alternating path of type II, a contradiction.

\begin{figure}[htb]
	\centering
	\includegraphics[page=7]{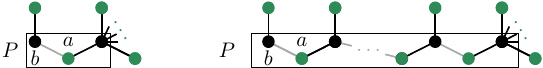}
	\caption{If vertices $a\in A_2$ and $b\in B_2$ are adjacent in $G$, there exists an alternating path of type~II.
	}
	\label{fig:alterPath}
\end{figure}

With arguments similar to the above, $(A_3,B_3)$ is a III-coloring; otherwise we find an alternating path of type III. Specifically, one can show that vertex $a$ belongs to a star $S_2$ and $b$ to a star $S_3$.
\end{proof}

\section{Conclusion}
We introduce the concept of scheduling with machine conflicts 
and hereby generalize the classical makespan minimization problem on identical parallel machines. We show inapproximability results and approximation algorithms for identical jobs. Moreover, we present a polynomial time algorithm for the case of unit jobs on bipartite graphs. This also induces a local optimality criterion for general graphs.

Various interesting avenues remain open for future research. In particular, the investigation of graph classes capturing geometric information is of special interest for applications in which spatial proximity implies  machines conflicts. Furthermore, closer investigation of  special cases of identical jobs and special graph classes can lead to a better understanding of structures of optimal schedules. Additionally, variants with preemption constitute natural interesting  directions for future research. We are positive that  our insights and tools prove to be useful.

%%
%% Bibliography
%%

\bibliography{./../bibliography}
%
%\newpage
%\appendix
%%
%\input{./99_appendix}

\end{document}